%% file: main.tex
\documentclass[conference]{IEEEtran}

\usepackage[noend, ruled]{algorithm2e}
\usepackage{amsmath}
\usepackage{amssymb}
\usepackage{amsthm}
\usepackage{float}
\usepackage{hyperref}
\usepackage{lipsum}
\usepackage{listings}
\usepackage{tikz}
\usepackage{txfonts}

\usetikzlibrary{calc, decorations.pathreplacing, math}
\newcommand*{\TickSize}{12pt}

\newenvironment{Figure}
  {\par\medskip\noindent\minipage{\linewidth}}
  {\endminipage\par\medskip}

\newtheorem{theorem}{Theorem}[section]
\newtheorem{lemma}[theorem]{Lemma}

\newtheorem{definition}[theorem]{Definition}

\newcommand{\bigO}{O}
\newcommand{\E}{\mathbf{E}}

\newcommand{\polylog}{\textrm{polylog}}

\lstdefinelanguage{swift}
{
morekeywords={
func,if,then,else,for,in,while,do,switch,case,default,where,break,continue,fallthrough,return,
typealias,struct,class,enum,protocol,var,func,let,get,set,willSet,didSet,inout,init,deinit,extension,
subscript,prefix,operator,infix,postfix,precedence,associativity,left,right,none,convenience,dynamic,
final,lazy,mutating,nonmutating,optional,override,required,static,unowned,safe,weak,internal,
private,public,is,as,self,unsafe,dynamicType,true,false,nil,Type,Protocol,
},
morecomment=[l]{//}, 
morecomment=[s]{/*}{*/}, 
morestring=[b]" 
}

\definecolor{keyword}{HTML}{BA2CA3}
\definecolor{string}{HTML}{D12F1B}
\definecolor{comment}{HTML}{008400}

\LinesNumbered
\DontPrintSemicolon
\allowdisplaybreaks
\let\footnotesize\small 
\def\paragraph#1{\vspace{0.25em}\noindent {\bf #1}} 

\title{Asynchronous Gossip in Smartphone \\ Peer-to-Peer Networks}

\author{
\IEEEauthorblockN{Calvin Newport}
\IEEEauthorblockA{\textit{Georgetown University} \\
Washington, DC, USA \\
cnewport@cs.georgetown.edu}
\and
\IEEEauthorblockN{Alex Weaver}
\IEEEauthorblockA{\textit{Georgetown University} \\
Washington, DC, USA \\
aweaver@cs.georgetown.edu}
\and
\IEEEauthorblockN{Chaodong Zheng}
\IEEEauthorblockA{\textit{Nanjing University} \\
Nanjing, China \\
chaodong@nju.edu.cn}
}

\begin{document}

\maketitle


\begin{abstract}
In this paper, we study gossip algorithms
in communication models that describe the peer-to-peer networking functionality included in most standard smartphone operating systems. 
We begin by describing and analyzing a new synchronous gossip algorithm in this setting 
that features both a faster round complexity and simpler operation than the best-known existing solutions.
We also prove a new lower bound on the rounds required to solve gossip that resolves a minor open question
by establishing that existing synchronous solutions are within logarithmic factors of optimal.
We then adapt our synchronous algorithm to produce a novel gossip strategy for an asynchronous
model that directly captures the interface of a standard smartphone peer-to-peer networking library
(enabling algorithms described in this model to be easily implemented on real phones).
Using new analysis techniques, we prove that this asynchronous strategy efficiently solves gossip.
This is the first known efficient asynchronous information dissemination result for the smartphone
peer-to-peer setting.
We argue that our new strategy can be used to implement effective information spreading
subroutines in real world smartphone peer-to-peer network applications,
and that the analytical tools we developed to analyze it can be leveraged
to produce other broadly useful algorithmic strategies for this increasingly important setting.
\end{abstract}




\maketitle

\section{Introduction}

\input{intro.tex}

\section{Preliminaries}
\label{sec:prelim}

\input{prelim.tex}

\section{Synchronous Gossip}
\label{sec:synch}
\input{mtm.tex}

\section{Asynchronous Gossip}
\label{sec:asynch}

\input{amtm.tex}

\bibliography{main}
\bibliographystyle{plainurl}

\appendix

\input{appendix.tex}

\end{document}

%% file: intro.tex
In this paper, we study gossip in smartphone peer-to-peer networks,
an interesting emerging networking platform that makes use of the peer-to-peer libraries
 included in standard smartphone operating systems (for examples of these networks in practice, see: \cite{suzuki2012soscast,aloi2014spontaneous,reina2015survey,lu2016networking,holzer2016padoc,firechat,oghostpot}). 
We begin by improving the best-known synchronous gossip algorithms in this setting,
and then build on these results to describe and analyze the first efficient
asynchronous solution. The model in which we study this latter algorithm
 captures the interfaces and guarantees of an actual peer-to-peer networking library
used in iOS, meaning that our gossip solution can be directly implemented on commodity iPhones.
To emphasize this practicality, in Appendix~\ref{apx:implementation} we provide the SWIFT
code that implements this algorithm in iOS---a rare
instance in the study of distributed algorithms for wireless networks in which the gap
between theory and practice is minimal.

Below we briefly summarize the models we study and the relevant existing bounds in these models,
before describing the new results proved in this paper.

\paragraph{The Mobile Telephone Model (MTM).}
The \emph{mobile telephone model} (MTM)~\cite{ghaffari:2016} extends the well-studied 
\emph{telephone model} of wired peer-to-peer networks (e.g.,\cite{telephone1,telephone2,telephone3,chierichetti2010rumour,fountoulakis2010rumor,giakkoupis2014tight})
to better capture the dynamics of the peer-to-peer network libraries implemented in existing smartphone operating systems.
In recent years, multiple distributed algorithm problems have been studied in the MTM setting,
including:  rumor spreading \cite{ghaffari:2016}, load balancing \cite{dinitz:2017}, leader election \cite{newport:2017},
network capacity \cite{dinitz2019capacity}, and gossip \cite{newport:2017b,newport2019random}.

As we elaborate in Section~\ref{sec:synch},
in the MTM, time proceeds in synchronous rounds.
At the beginning of each round, each wireless device (which we will call a {\em node}) 
can {\em advertise} a small amount
of information to its neighbors in the peer-to-peer network topology (defined by an undirected graph).
After receiving advertisements, nodes can attempt local connections.
In more detail, in each round, each node can send and accept at most one {\em connection proposal}.
If a node $u$'s proposal is accepted by neighboring node $v$,
then $u$ and $v$ can perform a bounded amount of reliable communication using this connection before the round ends.

This {\em scan-and-connect} network architecture---in which nodes can broadcast small advertisements to all of their neighbors,
but form pairwise connections with only a limited number at a time---is a defining feature of existing smartphone peer-to-peer libraries.
In the peer-to-peer libraries that depend on Bluetooth, for example, the advertisements are implemented as low energy beacons that contain 
at most tens of bytes, whereas the pairwise connections are implemented as reliable, high throughput links that can
achieve up to 2 Mbits/sec~\cite{barry:2012}. These libraries, therefore, allow devices to broadcast advertisements to all
neighbors, but severely restrict the number of concurrent 
pairwise connections allowed. In iOS, for example, this limit is $7$ (the MTM typically reduces this bound to $1$
to simplify the model description and analysis). 

\paragraph{Mobile Telephone Model vs.~Classical Telephone Model.}
The MTM can be understood as a modification of the classical telephone model of peer-to-peer networks~\cite{telephone1,telephone2,telephone3,chierichetti2010rumour,fountoulakis2010rumor,giakkoupis2014tight}.
The MTM differs from its predecessor in two ways: (1) it allows nodes to broadcast small advertisements to all neighbors;
and (2) it bounds the numbers of concurrent connections allowed at each node.
As elaborated in~\cite{kuhn:bounded,ghaffari:2016},
this second difference prevents existing telephone model results from applying to the mobile telephone setting,
as the best-known telephone model analyses specifically depend on the ability of nodes to service an unbounded
number of incoming concurrent connections (the standard analysis of PUSH-PULL rumor spreading, for example, 
depends on the ability of many nodes to simultaneously pull the rumor from a common informed neighbor).
On the other hand, the addition of advertisements to the MTM means that results in this new model
do not apply to the classical telephone setting, which not include this behavior.
Fundamentally new techniques are needed to study the MTM.

\paragraph{The Asynchronous Mobile Telephone Model (aMTM).}
The mobile telephone model includes synchronized rounds.
This assumption simplifies analyses that probe the fundamental capabilities of scan-and-connect style peer-to-peer networks.
It also introduces, however, a gap between theory and practice,
as real smartphone peer-to-peer networks are not synchronized.
To help close this gap, in~\cite{newport2019random}, the authors introduced
the {\em asynchronous mobile telephone model} (aMTM),
which, as we elaborate in Section~\ref{sec:asynch}, eliminates the synchronous round assumption from the MTM,
and allows communication events to unfold with unpredictable delays, controlled by an adversary.
To increase the practicality of the aMTM,
the authors of~\cite{newport2019random} also provide a software
wrapper around the network libraries offered in iOS
that matches the interface from the formal specification
of the aMTM---simplifying the task of directly implementing algorithms analyzed in the aMTM on iPhones.


\paragraph{Existing Results.}
Work on the MTM began with~\cite{ghaffari:2016}, which studied rumor spreading,
and described a strategy that uses a 1-bit advertisement to compensate for
connection bounds to spread a rumor in 
at most  $O((1/\alpha)\log^2{n}\log^2{\Delta})$ rounds, with high probability,
in a network with $n$ nodes, maximum degree $\Delta$, and vertex expansion $\alpha$ (see Section~\ref{sec:prelim}).
The paper also proved that there exist graphs with good graph conductance, $\phi$, 
for which efficient rumor spreading is impossible.
This creates a separation from the classical telephone model where {\em both} vertex expansion and conductance
are known to be good measures of the ability to spread a rumor efficiently in a graph.
In the classical model, for example, the canonical PUSH-PULL rumor spreading
strategy requires $\Theta((1/\alpha)\log^2{n})$ rounds for graphs with vertex expansion $\alpha$ \cite{giakkoupis2014tight},
and $\Theta((1/\phi)\log{n})$ rounds for graphs with conductance $\phi$ \cite{telephone2}. 

The more general problem of gossiping $k$ rumors in the mobile telephone model
was first studied in~\cite{newport:2017b},
which described an algorithm that spreads the rumors in $O((k/\alpha)\log^5{n})$ rounds,\footnote{In~\cite{newport:2017b}, the algorithm
is listed as requiring $O((k/\alpha)\log^6{n})$ rounds, but that result assumes a single bit advertisements in each round---requiring devices to spell out
control information over many rounds of advertising. To normalize with this paper,
in which tags can contain $\log{n}$ bits, this existing strategy's 
time complexity improves by a $\log$ factor.} with high probability.
This algorithm was {\em one-shot}, in the sense that it cannot accommodate on-going rumor arrivals, or detect when it has terminated.
In recent work~\cite{newport2019random},
a simpler gossip algorithm was described and analyzed that improves this bound to $O((k/\alpha)\log^2{n}\log^2{\Delta})$ rounds,
and can handle on-going rumor arrivals.

By comparison, the best-known gossip solution in the classical telephone model
requires $O(D+\polylog{(n)})$ rounds~\cite{censorhillel2017}. 
This result was considered a breakthrough as it removed the dependence on graph properties
such as expansion or conductance. The solution in~\cite{censorhillel2017}, however,
requires unbounded concurrent connections and unbounded message size (allowing all rumors
in the set difference between two nodes to be delivered during a given one-round connection\footnote{This explains why
the rumor count, $k$, is not needed in the time complexity}). 

The aMTM was introduced in~\cite{newport2019random},
which analyzes a basic asynchronous rumor spreading algorithm,
and prove it requires $O(\sqrt{(n/\alpha)}\cdot \log^2{n\alpha} \cdot \delta_{max})$ time, with high probability,
where $\delta_{max}$ is a sum of the maximum delays on the relevant communication events
(as is standard in asynchronous models, $\delta_{max}$ is unknown to the algorithm
and can change from execution to execution).
For gossip, however,
the paper establishes only a crude deterministic bound of $O(n\cdot k\cdot \delta_{max})$ time to gossip $k$ rumors.
Finding an efficient gossip algorithm in the aMTM was left as the core open question of~\cite{newport2019random},
as such an algorithm could be directly deployed as an information spreading routine in real smartphone peer-to-peer networks.

\paragraph{New Result \#1: Improved Synchronous Gossip.}
Our ultimate goal in this paper is to design and analyze an efficient and simple gossip strategy
for the aMTM. The first step toward this goal is to identify an efficient synchronous strategy that can
be adapted to asynchrony. The existing synchronous gossip algorithm from~\cite{newport2019random}
is not a good candidate for this purpose because it requires nodes to advertise
whether or not they were involved in a connection at any point during the previous $\log{n}$ rounds.
This behavior cannot be easily adapted to an environment with no rounds.

In Section~\ref{sec:synch},
we overcome this issue by describing a simpler strategy we call {\em random diffusion gossip} that does not depend on round history.
This algorithm has each node continually advertise two pieces of information
about its current rumor set: a hash of the set and its size. 
When faced with multiple neighbors with different rumor set hash values,
a node will randomly select a recipient of a connection proposal from among those with the smallest rumor set sizes.
This strategy is easily adapted to asynchrony as it does not explicitly use rounds.

As we show, 
in addition to being both round-independent and pleasingly straightforward in its operation, 
random diffusion gossip is more efficient than the solution from~\cite{newport2019random},
requiring only $O((k/\alpha)\log{n}\log^2{\Delta})$ rounds to spread $k$ rumors.
The source of this speed-up is a new and improved version of the core technical lemma from~\cite{ghaffari:2016},
which bounds the performance of a random matching strategy in bipartite graphs.
Notice that this gossip result also improves the best known result for rumor spreading (i.e., for $k=1$).

Finally, we note that these synchronous gossip bounds are all of the form $\tilde{O}(k/\alpha)$ (where $\tilde{O}$ suppresses
polylogarithmic factors in $n$ and $\Delta$). As argued in the previous work on gossip, 
it might be possible to leverage pipelining to achieve results in $\tilde{O}(k + (1/\alpha))$,
which would make the existing gossip strategies for this model far from optimal in certain cases.
In Section~\ref{sec:synch}, we resolve this open question by proving that $\Omega(k/\alpha)$ is indeed
a lower bound for spreading $k$ rumors in the mobile telephone model.

\paragraph{New Result \#2: Asynchronous Gossip.}
Our synchronous random diffusion gossip algorithm's operation is easily adapted to our asynchronous model.
Adapting its analysis, however, is more complicated.
Like most algorithms studied in the MTM, 
our synchronous analysis of random diffusion gossip relies on the
synchronized behavior of the devices in the network:
fixing for each round a set of potentially productive connections,
and then arguing that a reasonable fraction of these connections will succeed in parallel during the round.

Our first step toward enabling an asynchronous analysis is to divide time into
{\em intervals} of a length proportional to $\delta_{max}$.
These phases are not used by the algorithm (as $\delta_{max}$ is {\em a priori} unknown),
but instead meant only to facilitate our analysis.
As in the synchronous setting,
we fix a set of potential connections at the beginning of each interval.
We show that amidst all the chaotic, asynchronous behavior
that occurs during the interval, for each such connection from some node $u$ to some node $v$ in this set,
one of two things will happen: there will be a point at which $u$ selects a connection from a set that includes
$v$ and that is not too large (keeping the probability of $v$'s selection reasonable), or some other
node will end up connecting with $v$ before $u$ even gets a chance to learn about $v$.

To make use of this probabilistic analysis, we leverage a
rebuilt version of the core randomized matching lemma from~\cite{ghaffari:2016} (discussed above),
that we make not only more powerful but also significantly more friendly to asynchrony. 
In more detail, this new version includes two crucial changes.
First, the original lemma follows the behavior of
a randomized matching strategy over multiple rounds to achieve the needed result.
Our new version, by contrast, requires only a single round, which is necessary to apply to our interval structure, 
as in the asynchronous setting too much can change in the network between intervals to enable a coherent multi-interval graph analysis.
Second, the original version relied on the precise probabilities of particular connections occurring, using both
upper and lower bounds on these values to prove its claim.
Our new version only requires the loose lower bounds on connection probabilities established
by our asynchronous analysis.

Combining these techniques, we are able to
translate the synchronous complexity bound directly to the asynchronous setting,
proving that $k$ rumors spread in at most $O((k/\alpha)\log{n}\log^2{\Delta}\cdot\delta_{max})$ time.

%% file: prelim.tex
Here we define useful notation and results that we use throughout the analysis that follows.

\paragraph{Range Notation.}
We use the notation $[m]$ for $1\leq m$ to signify the range of integers $1,\ldots,\lceil m\rceil$. In contrast, we use the notation $[a,b]$, for $a \leq b$, to denote the real numbers from $a$ to $b$.

\paragraph{Graphs and Vertex Expansion.}
Fix an undirected graph $G=(V,E)$.
For node $u\in V$, we  use the notation $N(u)$ to denote $u$'s neighbors in $G$ and $deg(u)=|N(u)|$ to denote $u$'s degree in $G$. 
Let $\Delta=\max_{u\in V}deg(u)$ be the maximum degree of any node in $G$. For a given subset of nodes $S\subseteq V$, let $\partial S=\{v\mid v\in V\setminus S, N(v)\cap S\neq \emptyset\}$ denote the \emph{boundary} of $S$. We then let $\alpha(S)=|\partial S|/|S|$ and define the \emph{vertex expansion} of a graph $G$ as $\alpha =\min_{S\subset V,|S|\in[n/2]}\alpha(S)$.

\par

Let $B(S)$ represent a bipartite graph with bipartitions $(S, V\setminus S)$ and let $v(B(S))$ represent the size of the maximum matching over $B(S)$. We leverage the following lemma from \cite{ghaffari:2016}.

\begin{lemma}
\label{lem:matching}

(Lemma 5.4 of \cite{ghaffari:2016}). Let $\gamma=\min_{S\subset V, |S|\in[n/2]}\{\varv(B(S))/|S|\}$. It follows that $\gamma\geq\alpha/4$.

\end{lemma}

\paragraph{Useful Probability Results.}
Many of our results are described as holding {\em with high probability} (or, w.h.p.), which
we define to mean with a failure probability polynomially small in the network size $n$.
To help achieve these results, we sometimes apply concentration bounds,
often using the following presentation of the Chernoff bound.

\begin{theorem}
\label{thm:chernoff}

Let $X_1,\ldots,X_n$ be a series of independent random variables such that $X_i\in[0,1]$ where  $X =\sum_{i=1}^nX_i$ has expectation $\E[X]=\mu$. For $\varepsilon\in [0,1]$, $\Pr[X\leq(1-\varepsilon)\cdot\mu]\leq\exp(-(1/2)\cdot\varepsilon^2\mu)$.

\end{theorem}

In several places in our analysis, we tame correlated random variables 
by applying the following stochastic dominance result. This general idea is common,
but we prove the result from scratch here in the exact form we need
for the sake of completeness. The full proof resides in Appendix \ref{proof:stochasticdominance}.

\begin{lemma}
\label{lem:stochasticdominance}

Let $X_1,\ldots,X_T$ be a sequence of $T$ random indicator variables where $X_i= 1$ with some unknown probability $q_i$.
Assume $\forall{i}\in[T]$, it always holds that $q_i\geq p$, for some constant probability $p$. Next, define the total number of successes as $Y=\sum_{i\in[T]}X_i$.
It follows that $Y=\Omega(pT)$ with probability at least $\Omega(1-\exp({-pT}))$.

\end{lemma}


%% file: mtm.tex
In this section, we analyze new upper and lower bounds for
gossip in the synchronous MTM.

\subsection{The Mobile Telephone Model}

The mobile telephone model (MTM) (introduced in~\cite{ghaffari:2016}) describes a peer-to-peer network of wireless devices. The network is modeled as an undirected graph $G=(V,E)$, where each device $u$ is represented by a vertex in the graph. We will use the term \emph{node} to refer to both the device and the corresponding vertex in the graph. If two devices $u$ and $v$ are within communication range in the network, we connect the corresponding nodes with an undirected edge $\{u,v\}\in E$. We denote the number of nodes in the graph as $n=|V|$.

\par

Time in the MTM proceeds in synchronous rounds with all nodes beginning at round $1$. In each round, each node begins 
by broadcasting an advertisement containing $O(\log{n})$ bits to its neighbors in $G$.
After receiving advertisements,
each node can decide to send a {\em connection proposal} to at most one neighbor.
Any node that receives one or more proposals must accept exactly one. 
We allow the model to arbitrarily select which proposal is accepted in this case. 
(That is, we do not
necessarily assume that each node successfully receives {\em all} incoming proposals and is therefore able
to make a careful decision on which to accept.)

\par

Finally, if some node $v$ accepts a connection proposal from neighboring node $u$, 
then $u$ and $v$ are considered {\em connected}.
They can then perform a bounded amount of interactive and reliable communication before the round concludes.
Notice, this model definition limits each node to participating in at most $2$ connections per round (one outgoing and one incoming).

\subsection{The Gossip Problem}
The gossip problem we study assumes
that $k \geq 1$ gossip rumors (also called {\em tokens} in the following) are distributed arbitrarily 
to nodes at the beginning of the execution (that is, some nodes can start with many tokens, some can start with none).
The problem is solved once all nodes know all $k$ rumors.
Nodes do not know $k$ in advance.
We treat the gossip tokens as comparable black boxes.
The only way for a node $u$ to communicate a token to node $v$ is if $u$ and $v$ are connected.
In the synchronous setting, 
we limit nodes to communicating at most a constant number of tokens over a given connection in a single round.
(Later, when we study this problem in the asynchronous setting, we instead bound the maximum time required
to transmit a single token over a connection.)


\subsection{The Random Diffusion Gossip Algorithm}

Here we present the \emph{random diffusion gossip} algorithm which we formalize as pseudocode in Algorithm \ref{alg:synchalg}. The core strategy of this algorithm is for nodes to attempt to send tokens to the neighbors with the smallest token sets.
This contrasts to the strategy of~\cite{newport2019random} in which nodes bias connection attempts toward neighbors
that have not participated in connections in recent rounds.


\par

\begin{algorithm}
\SetAlgoNoLine
\caption{Random diffusion gossip (for process $u$)}
\label{alg:synchalg}
$T\gets$ initial token set of $u$\;
$H\gets$ shared hash function\;

\While {$\tt{true}$}
{
$\tt{Advertise}(\langle$$H(T)$, $|T|$, $u$$\rangle)$\;
$A \gets \tt{ReceiveAdvertisements}()$\;
\;
$s\gets \min{(\{s_v\mid \langle h, s_v, *\rangle \in A, h\neq H(T) \})}$\;
$N\gets \{v\mid\langle h, s, v\rangle\in A, h\neq H(T) \}$\;
$v\gets$ node chosen randomly from $N$\;
\textit{(attempt to connect to $v$; if successful, send/receive a token from the set difference)}\;
}
\end{algorithm}

In more detail, in each round, each node $u$ advertises a \textit{hash} of its token set, 
the \textit{size} of its token set, and its unique identifier. \footnote{As in~\cite{newport2019random},
a couple of simplifying assumptions are made here. The first is that we avoid hash collisions in the executions
we consider, allowing us to make the reasonable assumption that different token set hashes always
indicate different token sets. We also make the pragmatic assumption that these hash vaues, as well as token set size counts,
fit within the $O(\log{n})$ bound on advertisements.} Node $u$ then considers advertisements from neighbors that advertised different token set hashes,
identifying the smallest token set size from this set. It randomly selects one of these nodes 
to send a connection proposal.
If the proposal is accepted, a token from the set difference is transferred, increasing at least
one of the two nodes' token sets.


\subsection{Analysis}

Our goal is to prove the following bound on the time complexity of this algorithm.

\begin{theorem}
\label{thm:syncalg}

With high probability in $n$, the random diffusion gossip algorithm solves gossip in $\bigO((k/\alpha)\log{n}\log^2{\Delta})$ rounds,
where $k$ is the number of initial tokens, $\alpha$ is the vertex expansion of the graph, $n$ is the size of the graph, and $\Delta$ is the maximum degree of the graph.

\end{theorem}

We begin by defining some useful notation. At the beginning of round $r$, let $T_u(r)$ be the token set of node $u$ and let $s_u(r)$ be the minimum token set size among $u$'s neighbors. Furthermore, for a fixed topology graph $G=(V,E)$,
let $N(u)$ be the neighbors of $u$ in $G$ and let $N_u(r)$ be the \emph{productive} neighbors for $u$ at the beginning of round $r$,
where we define $N_u(r)=\{v\mid v\in N(u),|T_v(r)|=s_u(r), H(T_u(r))\neq H(T_v(r))\}$. 

\par

For integer sizes $i\in 0,\ldots,k$; let $S_i(r)=\{v\mid v\in V, i=|T_v(r)|\}$ be the set of nodes that know exactly $i$ tokens at the beginning of round $r$. Next, let $n_i(r)=|S_i(r)|$ and $n^*_i(r)=\min{(n_i(r), n-n_i(r))}$. 

\par

We also define $i_{min}(r)=\min(\{i\mid i\in 0,\ldots,k]\mid n_i(r) > 0\})$ as the minimum token set size for which there is at least one node with exactly that many tokens. For convenience, let $S_{min}(r)=S_j(r)$ and $n^*_{min}(r)=n^*_j(r)$ for $j=i_{min}(r)$. Finally, we define $C(r)=|\{i\mid i\in 0,\ldots,k]\mid n_i(r)>0\}|$ as the number of token set sizes held by nodes.

\par

The approach we will take when proving our theorem statement is to bound how long any minimum token set size $i_{min}(r)$ can remain the minimum token set size. Since the minimum token set size can never decrease, this will then allow us to prove the total time complexity for our algorithm. For most of our analysis, we will focus on the connections between nodes in $S_{min}(r)$ and $V\setminus S_{min}(r)$. In order for this cut to exist though, clearly it must be the case that $C(r)>1$. Therefore we quickly handle the case where $C(r)=1$, the proof of which can be found in Appendix \ref{proof:samesize}.

\begin{lemma}
\label{lem:samesize}

Fix a round $r > 0$ such that $C(r) = 1$. Either $C(r+1) > 1$ or $i_{min}(r+1)>i_{min}(r)$.

\end{lemma}

The purpose of Lemma \ref{lem:samesize} is to simply establish that regardless of the minimum token set size, there are some nodes which quickly achieve a token set larger than the minimum number of tokens held by any node. This allows us to analyze the cut between these nodes in $V\setminus S_{min}(r)$ and the nodes that still possess exactly $i_{min}(r)$ tokens, $S_{min}(r)$. Productive connections made over this cut will provide nodes of $S_{min}(r)$ with new tokens, increasing their token set size, and shrinking $S_{min}(r)$. When no nodes remain, the minimum token set size must be larger than $i_{min}(r)$.

\par

Furthermore, note that if for some rounds $r_1$ and $r_2$ such that $C(r_1)>1$, $C(r_2)=1$, and $r_1<r_2$ it must be the case that every node in $S_{min}(r_1)$ has participated in a productive connection. Therefore we will continue our analysis with the assumption that $C(r)>1$ for each round $r$ we fix and revisit Lemma \ref{lem:samesize} in the proof of Theorem \ref{thm:syncalg}.

\smallskip

We continue by defining the \emph{productive subgraph} $G(r)$ of $G$ (defined with respect to a fixed round $r$) which defines all the connections which nodes might attempt to form in the given round $r$.

\begin{definition}
\label{def:prodsubgraph}

At the beginning of round $r>0$, define the \textbf{productive subgraph} $G(r)$ of the graph topology $G=(V,E)$ as the undirected graph $G(r)=(V,E(r))$ such that $E(r)=\{\{u,v\}\mid v\in N_u(r)\}$.

\end{definition}

For the purposes of our analysis, it will be sufficient to focus on a subgraph of the productive subgraph which only considers nodes in $S_{min}(r)$ and their neighbors.

\begin{definition}
\label{def:minprodsubgraph}

At the beginning of round $r>0$, define the \textbf{minimum productive subgraph} $G_{min}(r)$ as the undirected bipartite subgraph $G_{min}(r)=(L_{min}(r),R_{min}(r),E_{min}(r))$ such that

\begin{itemize}
    \item $L_{min}(r) = \{u\mid u\in V\setminus S_{min}(r),N(u)\cap S_{min}(r)\neq \emptyset\}$
    \item $R_{min}(r) = \{u\mid u\in S_{min}(r),N(u)\cap (V \setminus S_{min}(r))\neq \emptyset\}$
    \item $E_{min}(r) = \{\{u,v\}\mid u\in L_{min}(r), v\in R_{min}(r), \{u,v\}\in E(r)\}$
\end{itemize}

\end{definition}

In other words, the minimum productive subgraph $G_{min}(r)$ only contains edges representing the potential connections which would result from connection proposals sent to nodes with the fewest number of tokens in the entire network at the beginning of round $r$ (from nodes with more than this number of tokens). The significance of $G_{min}(r)$ is that every productive connection in this graph causes a node with the fewest number of tokens to no longer have the fewest number of tokens. For this reason, we next lower bound the number of potential productive connections in $G_{min}(r)$. The full proof for this lemma can be found in Appendix \ref{proof:minmatching}.
\begin{lemma}
\label{lem:minmatching}

For a fixed round $r>0$, there is a matching over $G_{min}(r)$ with size $m\geq (\alpha/4)\cdot n^*_{min}(r)$.

\end{lemma}

We now have a lower bound for the number of potential connections that nodes in $S_{min}(r)$ could participate in for  a given round. To show that our algorithm is able exploit these possible connections, 
we now prove and apply a significantly reworked version of a core lemma from \cite{ghaffari:2016}
which bounds the behavior of randomized connection attempts in bipartite graphs satisfying certain
properties. In the immediate context of our synchronous analysis,
this new version of the lemma provides a log-factor time complexity improvement
as compared to the original version.
As detailed in the introduction, however, most of the updates captured below (which represent
some of the core technical contributions of this paper)
are introduced to
make this lemma applicable to the asynchronous analysis that follows in the next section.

We also note that that this improved version of the lemma can be plugged into the
analysis of \cite{ghaffari:2016} to provide a log factor improvement to the
complexity of its rumor spreading algorithm.

\begin{lemma}
\label{lem:prodconns}

(Replaces Theorem 7.4 in \cite{ghaffari:2016}). Let $G(L,R)$ be the subgraph of $G_{min}(r)$ induced by node subsets $L$ and $R$ and let $N_{L,R}(u)$ is the neighbors node $u$ in $G(L, R)$ and $deg_{L,R}(u)=|N_{L,R}(u)|$. Fix any $i\in [32\cdot\log{\Delta}]$. For a fixed round $r>0$, let $L\subseteq L_{min}(r)$ and $R\subseteq R_{min}(r)$ be subsets such that:
\begin{enumerate}
    \item there is a matching of size $|L|$ over $G(L,R)$,
    \item $|R|\geq |L|\geq c\cdot m$ for some $0<c\leq 1$ where $m$ is the size of the maximum matching over $G(L,R)$,
    \item $\sum_{u\in L}deg_{L,R}(u) \leq m\Delta^{1-\frac{i-1}{32\cdot\log{\Delta}}}$, and
    \item for every $u\in L$, every neighbor of $u$ in $R_{min}(r)$ is in $R$.
\end{enumerate}
With at least constant probability within one round of the random diffusion gossip algorithm,

\begin{enumerate}
    \item At least $\Omega\Big(\frac{m}{\log{\Delta}}\Big)$ nodes of $R$ participate in a productive connection, or
    \item We can identify $L''\subseteq L\cap L_{min}(r')$ and $R''\subseteq R\cap R_{min}(r')$ for some $r'\in\{r, r+1\}$ such that:
    \begin{enumerate}
        \item  there is a matching of size $|L''|$ over $G(L'', R'')$,
        \item $|R''|\geq |L''|\geq (1-1/\log{\Delta})^2\cdot|L|$,
        \item $\sum_{u\in L''}deg_{L'',R''}(u)\leq m\Delta^{1-\frac{i}{32\cdot\log{\Delta}}}$, and
        \item for every $u\in L''$, every neighbor of $u$ in $R_{min}(r')$ is in $R''$.
    \end{enumerate}
\end{enumerate}

\end{lemma}

\begin{proof}

Our proof, like the proof of Theorem 7.4 in \cite{ghaffari:2016}, is broken up into several steps. For the matching $M$ of size at least $m\cdot c$ over our original graph $G(L,R)$, we denote a node $v\in R$ as the \emph{original match} of a node $u\in L$ if $\{u,v\}\in M$. This terminology is also taken from the original proof.

\paragraph{Remove High Degree Nodes from \texorpdfstring{$L$}{L}.} Let $\delta_i=(1/c)\cdot\log{\Delta}\cdot\Delta^{1-\frac{i-1}{32\cdot\log{\Delta}}}$ and consider all nodes in $L$ with degree at most $\delta_i$. As in \cite{ghaffari:2016} this choice of $\delta_i$ is based on our assumptions that $|L|\geq c\cdot m$ and $\sum_{u\in L}deg_{L,R}(u)\leq m\Delta^{1-\frac{i-1}{32\cdot\log{\Delta}}}$ such that at most a $1/\log{\Delta}$ fraction of the nodes $u\in L$ can have $deg_{L,R}(u)>\delta_i$. Let $L'\subseteq L$ be the subset of nodes once we remove all such high degree nodes from $L$ and again note that $|L'|\geq (1-1/\log{\Delta})\cdot|L|$.

\par

We then remove all nodes from $R$ that are not connected to $L'$ and denote the remaining set $R'$. Note that for every node $u\in L'$, every neighbor $N_{L,R}(u)=N_{L',R'}(u)$. The authors of \cite{ghaffari:2016} note that this implies $G(L',R')$ has a matching of size $|L'|$ since for every node $u\in L'$, $u$'s original match is in $R'$. These observations alone fulfill conditions $a$, $b$ and $d$ of the second objective of the lemma. Therefore if condition $c$ holds such that $\sum_{u\in L'}deg_{L',R'}(u)\leq m\Delta^{1-\frac{i}{32\cdot\log{\Delta}}}$, the second objective of the lemma is already satisfied by setting $L''=L'$, $R''=R'$, and $r'=r$. We therefore assume for the remainder of the proof that $\sum_{u\in L'}deg_{L',R'}(u)\geq m\Delta^{1-\frac{i}{32\cdot\log{\Delta}}}$.

\par

At this point we diverge significantly from the strategy of the original proof and introduce a new technique for leveraging this assumption regarding the the degree sum in $G(L',R')$. We start by leveraging a definition which was first used in \cite{alon:1986} in the context of the maximal independent set problem. Namely, call a node $u$ \emph{good}  with respect to a graph $G$ if $|\{v|v\in N(u),deg(v)\leq deg(u)\}|\geq deg(u)/3$ where $N(u)$ and $deg(u)$ are $u$'s neighbor set and degree in $G$. Otherwise call $u$ \emph{bad} with respect to $G$. In other words, a node is good with respect to a graph $G$ if at least one third of its neighbors in $G$ have at most its degree in $G$.

\par

In $G(L',R')$ let $R'_b\subseteq R'$ be the bad nodes in $R'$ and let $R'_g\subseteq R'$ be the good nodes, where good and bad are defined with respect to $G(L',R')$. Since every edge $G(L',R')$ has an endpoint in $R'$, clearly either $\sum_{u\in R'_b}deg_{L',R'}(u)\geq (1/2)\cdot m\Delta^{1-\frac{i}{32\cdot\log{\Delta}}}$ or $\sum_{u\in R'_g}deg_{L',R'}(u)\geq (1/2)\cdot m\Delta^{1-\frac{i}{32\cdot\log{\Delta}}}$. Simply speaking since $R'_b\cup R'_g=R'$, at least half of the edges in $G(L',R')$ are incident on $R'_b$ or at least half are incident on $R'_g$. We first assume the former case.

\paragraph{Case \#1: At Least Half the Edges in \texorpdfstring{$G(L',R')$}{G(L',R')} are Incident on \texorpdfstring{$R'_b$}{R'b}.}  Let $G(L'_b, R'_b)$ be the graph induced by the edges incident on $R'_b$ and note that for every $v\in R'_b$, $N_{L',R'}(v)=N_{L'_b,R'_b}(v)$. Next, recognize that if for any $v\in R'_b$, $deg_{L'_b,R'_b}(v)>\delta_i$, $v$ would have higher degree in $G(L',R')$ than any node in $L'$ (since every node in $L'$ has degree at most $\delta_i$), making $v$ trivially good with respect to $G(L',R')$. This contradicts $v\in R'_b$, therefore $deg_{L'_b,R'_b}(v)\leq \delta_i$.

\par

Divide the nodes of $R'_b$ into $\lceil\log{\delta_i}\rceil$ classes based on their degree in $G(L'_b,R'_b)$ such that nodes of degree $[2^{j-1},2^{j}]$ are in the $j$th class, denoted $R'_b(j)$. Let $E_j$ be the edges incident on nodes of this class. Note that by our case assumption, $\sum_{j\in [\log{\delta_i}]}|E_j|\geq (1/2)\cdot m\Delta^{1-\frac{i}{32\cdot\log{\Delta}}}$. Since for every node in $u\in L'_b$, $deg_{L'_b,R'_b}(u)\leq deg_{L',R'}(u)\leq \delta_i$, the probability that $v\in R'_b(j)$ participates in a productive connection in round $r$ is at least

\begin{align}
    &1-\prod_{u\in N_{L'_b,R'_b}(v)}\Bigg(1-\frac{1}{deg_{L'_b,R'_b}(u)}\Bigg)\\
    &\geq 1-\prod_{u\in N_{L'_b,R'_b}(v)}\Bigg(1-\frac{1}{deg_{L',R'}(u)}\Bigg)\\
    &\geq 1-\prod_{u\in N_{L'_b,R'_b}(v)}\Bigg(1-\frac{1}{\delta_i}\Bigg)
    = 1-(1-1/\delta_i)^{deg_{L'_b,R'_b}(v)}\\
    &\geq 1-(1-1/\delta_i)^{2^{j-1}}
    \geq 1-\frac{1}{1+2^{j-1}/\delta_i}\label{ineq1}\\
    &= \frac{2^{j-1}/\delta_i}{1+2^{j-1}/\delta_i}\geq \frac{2^{j-1}}{2\delta_{i}}\label{ineq2}
\end{align}

Note that for Line \ref{ineq1} we use the inequality $(1+x)^n\leq \frac{1}{1-xn}$ for $x\in[-1,0],n\in \mathbb{N}$ (which can be shown via Bernoulli's inequality) and for Line \ref{ineq2} we use our observation that $2^{j-1}\leq \delta_i$ and therefore $2^{j-1}/ \delta_i\leq 1$. Now, since there are $|E_j|$ edges incident on nodes in $R'_b(j)$ and each node in $R'_b(j)$ has degree at most $2^{j}$, there are at least $\frac{|E_j|}{2^{j}}$ nodes in this class. Therefore, the expected number of nodes from $R'_b(j)$  which participate in a productive connection $r$ is at least
$\frac{|E_j|}{2^j}\cdot \frac{2^{j-1}}{2\delta_i}
=\frac{|E_j|}{4\delta_i}$.

Therefore, the expected number of nodes selected across all $\log{\delta_i}$ classes is
\setcounter{equation}{0}
\begin{align}\label{line:beforesub}
&\sum_{j\in[\log{\delta_i}]}\frac{|E_j|}{4\delta_i}=\frac{1}{4\delta_i}\sum_{j\in[\log{\delta_i}]}|E_j|\\
&\geq\frac{(1/2)\cdot m\Delta^{1-\frac{i}{32\cdot\log{\Delta}}}}{4\delta_i}
=\frac{1}{4\delta_i}\sum_{j\in[\log{\delta_i}]}|E_j|\\\label{line:aftersub}
&\geq\frac{(1/2)\cdot m\Delta^{1-\frac{i}{32\cdot\log{\Delta}}}}{4\delta_i}
=\frac{ m\Delta^{1-\frac{i}{32\cdot\log{\Delta}}}}{8\delta_i}\\
&=\frac{m \cdot\Delta^{1-\frac{i}{32\cdot \log{\Delta}}}}{8\cdot (1/c)\cdot\log{\Delta}\cdot \Delta^{1-\frac{i-1}{32\cdot\log{\Delta}}}}=\frac{m\cdot c\cdot
\Delta^{-\frac{1}{32\cdot\log{\Delta}}}}{8\cdot \log{\Delta}}
=\Omega\bigg(\frac{m}{\log{\Delta}}\bigg)
\end{align}

Please note that here we derive Line \ref{line:aftersub} from Line \ref{line:beforesub} by leveraging our case assumption. Since this expectation is equal to the sum of negatively-correlated random variables, as in \cite{ghaffari:2016} we can then apply the Chernoff bound from Theorem \ref{thm:chernoff} to achieve a concentration around this bound to show that the probability the actual number of productive connections (for example) is at most a $1/8$ fraction of this expectation is at most 0.69. Note that this is only true if $m/\log{\Delta}\geq 1$ but if $m/\log{\Delta}< 1$ then the lemma is satisfied by informing just a single node in $R'$ which happens trivially. Therefore, with at least constant probability the first objective of the lemma is satisfied.

\paragraph{Case \#2: At Least Half the Edges in \texorpdfstring{$G(L',R')$}{G(L',R')} are Incident on \texorpdfstring{$R'_g$}{R'g}.} Now assume $\sum_{u\in R'_g}deg_{L',R'}(u)\geq (1/2)\cdot m\Delta^{1-\frac{i}{32\cdot\log{\Delta}}}$. As before, let $G(L'_g, R'_g)$ denote the graph induced by the edges incident on $R'_g$ and note that for all $u\in L'_g$, $deg_{L'_g,R'_g}(u)\leq deg_{L',R'}(u)$ and for all $v\in R'_g$, $N_{L'_g,R'_g}(v)= N_{L',R'}(v)$. For $v\in R'_g$, let the notation $N^\ell_{L'_g,R'_g}(v)=\{u\mid u\in N_{L'_g,R'_g}(v), deg_{L'_g,R'_g}(u)\leq deg_{L'_g,R'_g}(v)\}$ denote $v$'s lower degree neighbors in $G(L'_g,R'_g)$. Define $N_{L',R'}^\ell(v)$ for $v\in R'$ similarly. Therefore, since for all $u\in L'_g$, $deg_{L'_g,R'_g}(u)\leq deg_{L',R'}(u)$ and for all $v\in R'_g$, $N_{L'_g,R'_g}(v)= N_{L',R'}(v)$, we have that for all $v\in R'_g$, $N^\ell_{L',R'}(v)\subseteq N^\ell_{L'_g,R'_g}(v)$. Therefore, the probability a node $v\in R'_g$ is selected is at least,

\setcounter{equation}{0}
\begin{align}
    &1-\prod_{u\in N_{L'_g,R'_g}(v)}\Bigg(1-\frac{1}{deg_{L',R'}(u)}\Bigg)\\
    &\geq 1-\prod_{u\in N^\ell_{L'_g,R'_g}(v)}\Bigg(1-\frac{1}{deg_{L',R'}(u)}\Bigg)\\
    &\geq 1-\prod_{u\in N^\ell_{L',R'}(v)}\Bigg(1-\frac{1}{deg_{L',R'}(u)}\Bigg)\\
    \label{line:lightneighbors}
    &\geq 1-\prod_{u\in N^\ell_{L',R'}(v)}\Bigg(1-\frac{1}{deg_{L',R'}(v)}\Bigg)\\
    &\geq 1-(1-1/deg_{L',R'}(v))^{|N^\ell_{L',R'}(v)|}\\
    \label{line:gooddef}
    &\geq 1-(1-1/deg_{L',R'}(v))^{deg_{L',R'}(v)/3}\geq 1-e^{-1/3}>1/4
\end{align}

\par

Note that for Line \ref{line:lightneighbors} we use our definition of $N^\ell_{L',R'}(v)$ to replace the degree of $u$ with that of $v$ and Line \ref{line:gooddef} is where we leverage the assumption that $v$ is good. Now remove every node from $R'$ that is selected in round $r$ and denote the remaining set $R''$, and remove from $L'$ every node $u$ for which $u$'s original match was removed from $R'$. Denote the remaining nodes $L''$. Since we know from the above that each node $v\in R'_g$ is removed from $R'$ with probability at least $1/4$ and the probability that any edge $\{u,v\}$ is removed from $G(L',R')$ is at least the probability that $v$ is removed from $R'$, the probability that an edge $\{u,v\}$ incident on $R'_g$ is removed is at least $1/4$. Since by our case assumption that there are at least $(1/2)\cdot m\Delta^{1-\frac{i}{32\cdot\log{\Delta}}}$ edges incident on nodes in $R'_g$, in expectation, at least $(1/8)\cdot m\Delta^{1-\frac{i}{32\cdot\log{\Delta}}}$ edges are removed from $G(L',R')$. Therefore, since by our initial assumption that $\sum_{u\in L'}deg_{L,R}(u)\leq m\Delta^{1-\frac{i-1}{32\cdot\log{\Delta}}}$ and the fact that $\sum_{u\in L'}deg_{L',R'}(u)\leq\sum_{u\in L}deg_{L,R}(u)$,  the expected number of edges $X$ remaining in $G(L'',R'')$ is at most
 $\mathbf{E}[X]\leq m\Delta^{1-\frac{i-1}{32\cdot\log{\Delta}}}-(1/8)\cdot m\Delta^{1-\frac{i}{32\cdot\log{\Delta}}}
$. Therefore, by applying Markov's inequality we can upper bound the probability that $X\geq m\Delta^{1-\frac{i}{32\cdot\log{\Delta}}}$:
\begin{align*}
    \mathbf{Pr}[X\geq  m\Delta^{1-\frac{i}{32\cdot\log{\Delta}}}]&\leq \frac{m\Delta^{1-\frac{i-1}{32\cdot\log{\Delta}}}-(1/8)\cdot m\Delta^{1-\frac{i}{32\cdot\log{\Delta}}}}{m\Delta^{1-\frac{i}{32\cdot\log{\Delta}}}}
    &= \Delta^{\frac{1}{32\cdot\log{\Delta}}}-1/8
    = 2^{1/32} - 1/8
    <15/16
\end{align*}
Therefore, we've shown that $\sum_{u\in L''}deg_{L'',R''}(u)\leq m\Delta^{1-\frac{i}{32\cdot\log{\Delta}}}$ with at least constant probability. Since this satisfies condition $c$ of the second objective of the lemma, we conclude by showing that either the remaining conditions of this objective are satisfied or the first objective has been achieved.
\par

If $|L''|< (1-1/\log{\Delta})^2\cdot |L|$ then note that this means $|L''|< (1-1/\log{\Delta})\cdot |L'|$ since $|L'|\geq (1-1/\log{\Delta})\cdot |L|$. As is noted in \cite{ghaffari:2016}, this implies that at least a $1/\log{\Delta}$ fraction of nodes in $L'$ had their original match removed in round $r$ which means that at least $|L'|/\log{\Delta}$ nodes of $R'$ were selected and therefore participated in a productive connection. Since $|L'| \geq (1-1/\log{\Delta})\cdot |L|=\Omega(m)$ this would indicate that at least $\Omega(m/\log{\Delta})$ nodes of $R'$ participated in a productive connection which would satisfy the first objective of the lemma.

\par

Therefore, assume $|L''|\geq (1-1/\log{\Delta})^2\cdot |L|$. Note that once again by our construction, for every node $u\in L''$, every neighbor of $u$ in $R_{min}(r+1)$ is in $R''$. This includes $u$'s original match in $R'$ such that there is a matching over $G(L'',R'')$ of size $|L''|$. Furthermore by our construction of $G(L'',R'')$, $L''\subseteq L\cap L_{min}(r+1)$ and $R''\subseteq R\cap R_{min}(r+1)$. Therefore, with at least constant probability, the second objective  is satisfied for $L''$, $R''$, and $r'=r+1$.
\end{proof}

We now apply Lemma \ref{lem:prodconns} inductively over $\bigO(\log{\Delta})$ rounds and leverage our result from Lemma \ref{lem:minmatching} to bound the number of connections over this period. This proof can be found in Appendix \ref{proof:fullprogress}.

\begin{lemma}
\label{lem:fullprogress}

Fix a round $r>0$.  With at least constant probability, within at most $O(\log{\Delta})$ rounds at least $\Omega((\alpha/\log{\Delta})\cdot n^*_{min}(r))$ nodes of $S_{min}(r)$ participate in a productive connection.

\end{lemma}

Now that we have bounded the expected number of successful connections over a phase of $\bigO(\log{\Delta})$ rounds, our goal will be to bound the number of rounds required to increase the minimum token set size in the entire network. To establish this lemma (for which the proof can be found in Appendix \ref{proof:finishsizeincrease}) we first leverage Lemma \ref{lem:sizeincrease} which only bounds the time required for at least half the nodes to have more than the minimum token set size.

\begin{lemma}
\label{lem:finishsizeincrease}

Fix a round $r>0$. There exists a round $r_t$ where $r_t= r+\bigO((1/\alpha)\log{n}\log^2{\Delta})$ such that w.h.p. in $n$, all nodes in $S_{min}(r)$ participate in a productive connection by round $r_t$.
\end{lemma}

\begin{proof}[Proof (of Theorem \ref{thm:syncalg})]

We now have everything we need to prove our main theorem. Consider any round $r$ with minimum token set size $i_{min}(r)$. From Lemma \ref{lem:samesize} we have that if $C(r)=1$ then in round $r+1$ either $C(r+1)>1$ or $i_{min}(r+1)>i_{min}(r)$. As we show in Lemmas \ref{lem:minmatching} through \ref{lem:finishsizeincrease}, if the former is true for each round $r'$ we consider where $r'>r$ and $C(r')>1$, at most $((1/\alpha)\log{n}\log^2{\Delta})$ total rounds are needed for every node in $S_{min}(r)$ to participate in a productive connection. Furthermore, if instead $C(r')=1$ for any such round then clearly it is still the case that every node of $S_{min}(r)$ has participated in a productive connection in this many rounds.

\par

Therefore this necessitates that the minimum token set size after $((1/\alpha)\log{n}\log^2{\Delta})$ is at least $i_{min}(r)+1$. Since clearly, the minimum token set size can increase at most  $k$ times, the total round complexity of the algorithm to spread all $k$ tokens is at most O$((k/\alpha)\log{n}\log^2{\Delta})$ total rounds.
\end{proof}

\subsection{Lower Bound for Gossip in the Mobile Telephone Model}

We now show that our algorithm is optimal to within polylogarithmic factors by proving a lower bound for gossip in our model. 
The proof can be found in Appendix \ref{proof:gossiplowerbound}.

\begin{theorem}
\label{thm:gossiplowerbound}

For $k$ initial tokens and any value $1/n\leq \alpha\leq 1/2$, there is a graph with $n$ nodes and vertex expansion at least $\alpha$ where $\Omega(k/\alpha)$ rounds are required to solve the gossip problem.
\end{theorem}

%% file: amtm.tex
Here we analyze an asynchronous version of our random diffusion gossip strategy
in the aMTM.

\subsection{The Asynchronous Mobile Telephone Model}

The asynchronous mobile telephone model (aMTM), first introduced in \cite{newport2019random}, 
removes the assumption of synchronous rounds from the MTM. 
Core communication properties, such as the time required for a neighbor to receive an advertisement or connection proposal,
or the speed at which information is transmitted over a connection, can now vary arbitrarily
during an execution.

Similar to the MTM, the topology of the underlying network is defined by an undirected graph.
Furthermore, the behavior of the nodes in the aMTM is similarly constrained by a fixed {\em scan-and-connect} behavioral loop in which nodes: 
update their own advertisement, wait to hear new advertisements from at least some neighbors, 
decide whether to act on these advertisements by
attempting to form a connection with a neighbor, then repeat. Unlike the MTM, however,  nodes do not progress through this loop in a synchronized manner, with delays decided by an adversarial scheduler. This loop is formalized in Algorithm \ref{alg:amtminterface}.
The model implements the methods ${\tt update}$, ${\tt receiveAds}$, and ${\tt blockForConn}$,
which abstract the details of the underlying asynchronous communication.

\begin{algorithm}
\SetAlgoNoLine
\caption{The aMTM interface (for device $u$)}
\label{alg:amtminterface}

$state \gets {\tt idle}$\;
\textsc{Initialize()}\;
\While{${\tt true}$}
{
$tag \gets$ \textsc{GetTag}()\;
${\tt update}(tag)$\;
$receiver\gets {\tt null}$\;
$A\gets {\tt receiveAds}()$\;
\lIf{$A \neq \emptyset$}
{
$receiver \gets$ \textsc{Select}($A$)
}
\label{alg:line:blockforadds}
\lIf{$receiver\neq {\tt null}$}
{
$state \gets {\tt blockForConn}(receiver)$
}
\If{$state = {\tt connected}$}
{
\textsc{Communicate}($receiver$)\;
$state \gets {\tt idle}$
}
}
\end{algorithm}

\paragraph{Model Guarantees and Parameters.}
As is standard with asynchronous network models,
we constrain the model's behavior with respect to a set of maximum delays and bit rates
specified for its key communication activities.
We define these delays for a given execution
with the parameters $\delta_{update}$, $\delta_{conn}$, and $R_b$.
The values of these parameters are not known to the algorithm and can change from execution to execution. 
We detail the guarantees they help specify below:

\smallskip

{\em Advertisement Guarantees:} 
If a node $u$ calls {\tt update} at some time $t$,
then the model guarantees that every neighbor of $u$ must receive
an advertisement from $u$ in the interval $t$ to $t + \delta_{update}$,
and that only advertisements $u$ passed to {\tt update} during this interval are received in this interval.
Notice, there is no guarantee that $u$'s neighbors receive {\em all} of its advertisements.
It is possible, for example, that $u$ advertises $a$ at a given time, then loops back 
around
in less than $t_{update}$ time and replaces this with a new advertisement $a'$ 
before any neighbor had a chance to receive $a$.
On the other hand, once a node begins advertising, its neighbors will hear from
it at least once every $\delta_{update}$ time.


\smallskip

{\em Connection Attempt Guarantees:}
The parameter $\delta_{conn}$ bounds the maximum time required
for the {\tt blockedForConn} model method to resolve a connection attempt and
return whether or not the attempt succeeded.
In more detail, when $u$ calls {\tt blockedForConn}($v$),
for some neighbor $v$,
the model guarantees to deliver a connection proposal
to $v$. If $v$ is already engaged in a connection (i.e., it has previously accepted
a proposal and the resulting connection is still open), it will reject $u$'s proposal.
Otherwise, it will accept the proposal. The model must deliver the proposal
and the response within this interval of length $\delta_{conn}$.
The loop blocks until this underlying communication completes and {\tt blockForConn}
can return the status of the connection.
Notice that as in the synchronous model,
these guarantees prevent any node from servicing more than one incoming connection at a time.

\smallskip

{\em Communication Guarantees:}
Assume $v$ accepts $u$'s connection proposal.
At this point, they are connected and can communicate as specified by their respective
\textsc{Communicate} methods.
For many algorithms, such as the gossip strategy studied in this paper,
we simply specify what occurs during this connection in the sender's \textsc{Communicate}
routine. When implementing  algorithms, however,
this behavior must be explicitly specified for both
the sender and receiver roles.
The amount of time required by these interactions depends on both the amount
of information transmitted by \textsc{Communicate} and the transmission rates
determined by the model.
We use the parameter $R_b$ to bound the {\em minimum} bit rate at which
the model can transmit information between a connected pair of neighbors.
We assume that when a call to \textsc{Communicate} returns,
the connection is closed. It follows that each node can participate in at most one outgoing connection
at a time.

\smallskip

\paragraph{Implementation.}
The authors in~\cite{newport2019random} provide an implementation of the aMTM interface
in iOS. This implementation works with the peer-to-peer networking libraries
included in iOS to execute the main aMTM loop.
The algorithm designer working with this interface need only
implement the \textsc{Initialize}, \textsc{Update},
\textsc{Select}, and \textsc{Communicate} functions.
This close connection between the abstract aMTM model and real world implementation
simplifies the task of deploying on iPhones any peer-to-peer algorithm described in the aMTM. To underscore the directness of this connection,
we provide in Appendix~\ref{apx:implementation}
the straightforward SWIFT code that implements our gossip strategy in iOS.

\subsection{The Asynchronous Random Diffusion Algorithm}

\par We now introduce our asynchronous random diffusion gossip algorithm,
which is formalized in Algorithm~\ref{alg:asyncalg}.
This algorithm adapts the strategy of synchronous random diffusion gossip
to the asynchronous setting. The major difference is that in each loop iteration, a node selects a neighbor for connection from the set of 
advertisements it has received since the last iteration.
In the synchronous setting, by contrast, a node is always considering the
latest advertisements from all of its neighbors.

\begin{algorithm}
\caption{Asynchronous random diffusion gossip (for node $u$)} 
\label{alg:asyncalg}

\SetAlgoNoLine

\SetKwProg{func}{function}{}{}
\SetFuncSty{textsc}

\SetKwFunction{init}{Initialize}
\func{\init{}}
{
$T\gets$ initial tokens (if any) known by $u$\;
$H\gets$ some hash function\;
}
\SetKwFunction{getTag}{GetTag}
\func{\getTag{}}
{
\Return $\langle H(T),|T|,u\rangle$
}
\SetKwFunction{select}{Select}
\func{\select{$A$}}
{
$\hat{s} \gets \min{(\{s\mid \langle h, s, *\rangle \in A, h\neq H(T)\})}$\;\label{line:asyncalg:s}
$\hat{N} \gets \{v\mid\langle h, \hat{s}, v\rangle \in A, h\neq H(T)\}$\;\label{line:asyncalg:n}
\Return node chosen randomly from $\hat{N}$
}
\SetKwFunction{communicate}{Communicate}
\func{\communicate{$v$}}
{
\textit{(send/receive a token in the set difference with $v$)}
}
\end{algorithm}

\subsection{Analysis}

We prove the following bound on the time complexity of asynchronous random diffusion gossip.
\begin{theorem}
\label{thm:asyncalg}

With high probability in $n$, the asynchronous random diffusion gossip algorithm solves the gossip problem in time $\bigO((k/\alpha)\log{n}\log^2{\Delta}\cdot\delta_{max})$, where where $k$ is the number of tokens, $n$ is the network size,
$\alpha$ is the vertex expansion of the network, $\Delta$ is the maximum degree, and $\delta_{max}$ upper bounds the time required for one iteration of the aMTM loop for this algorithm.
\end{theorem}

The $\delta_{max}$ parameter included in the above theorem was introduced
to simplify notation by eliminating the need to cite
multiple timing parameters in our complexity bound.
Formally, we define:  $\delta_{max}=\delta_{conn}+\delta_{update}+b_{max}/R_b$,
where $b_{max}$ describes
the maximum size (in bits) of a gossip token.

Notice, because our algorithm only transfers a constant number of
tokens in each call to \textsc{Communicate}, 
each such call requires at most $O(b_{max}/R_b)$ time.\footnote{We omit for now the time required 
for two connected nodes to determine {\em which} token to transfer. Our algorithm
simply specifies that they transfer {\em some} token in the set difference of their token sets.
For the sake of completeness, one could add an additional parameter to capture
the maximum bits needed to also decide on this set difference. We omitted this extra
parameter for now as in the application scenarios we envision, the token sizes are often
large enough their transfer swamps the overhead required to identify which token to transfer.
In the event that the token set sizes are allowed to become massive, however, we can leverage the token transfer subroutine from \cite{newport:2017b} to decide this set difference using only $\bigO(\polylog{(k)})$ additional bits.}
The $\delta_{update}$ and $\delta_{conn}$ parameters upper bound the time required
to get through the {\tt update} and {\tt blockedForConn} methods, respectively.
It follows that each iteration of our gossip algorithm's main aMTM loop
requires at most $O(\delta_{max})$ time, making $\delta_{max}$ a useful aggregate parameter
for bounding asynchronous time complexity.

\smallskip

For the analysis that follows, we re-purpose much of our notation and several of our definitions from the synchronous setting. We will accomplish this through a slight abuse of notation in which we take any element parameterized with an integer round $r$ in the previous section and redefine it with respect to a real time $t$. For example, let $T_u(t)$ be the token set of node $u$ at time $t$ in the same way $T_u(r)$ was $u$'s token set at the beginning of round $r$.

\par

Similarly, we can adapt our notions of the productive subgraph $G(t)$, and minimum productive subgraph $G_{min}(t)$, for a time $t$ using the values of $s_u(t)$ and $N_u(t)$. That being said, some additional care is required in dealing with these graphs in the asynchronous model.
In a round-based setting, you can fix the productive subgraph at the beginning of the round and know
that all nodes will make connection decisions based on that exact graph during the round.
In the asynchronous model no such guarantees hold. You might fix a productive subgraph
at some time $t$, for example, but that graph can change before all the nodes get a chance to learn
it and make a connection decision.

\par

To handle this nuance, we introduce our first pieces of notation unique to our asynchronous analysis. 
Fix some time $t$ at which some node $u$ calls \textsc{Select}.
Let $\hat{N}_u(t)$ and $\hat{s}_u(t)$ be the values calculated on Lines \ref{line:asyncalg:s} and \ref{line:asyncalg:n} of Algorithm \ref{alg:asyncalg}, respectively, during this call to \textsc{Select}.
These values are calculated from the advertisement set $A_u(t)$ which is passed to node $u$'s call to the \textsc{Select} function at time $t$. Note that for a particular time $t$ and node $u$, $N_u(t)$ and $\hat{N}_u(t)$, and $s_u(t)$ and $\hat{s}_u(t)$, can vary, as $N_u(t)$ and $s_u(t)$ are based
on the status of the network at exactly time $t$, whereas 
$\hat N_u(t)$ and $\hat s_u(t)$ are based on the advertisement set passed to \textsc{Select}
at time $t$ (which may by that point already be out of date).
Also note that $\hat N_u(t)$ and $\hat s_u(t)$ are undefined for times that do not correspond
to a \textsc{Select} call. 

\par

To help tame this reality that a given node's snapshot of the network can become out of date
before it has a chance to act on it, 
we introduce the following definition concerning snapshots of the changing minimum productive subgraph:

\begin{definition}

Fix a time interval $[t_1, t_2]$ and two nodes $u,v\in V$ such that $\{u,v\}\in E_{min}(t_1)$. We say that $u$ \textbf{properly considers} $v$ with respect to $t_1$ during this interval if there exists a time $t_{consider}$, $t_1\leq t_{consider}\leq t_2$, such that $u$ calls \textsc{Select} at $t_{consider}$,
 $v\in \hat{N}_u(t_{consider})$, and $|\hat{N}_u(t_{consider})|\leq deg_{min}(u)$,
where $deg_{min}(u)$ is the degree of $u$ in $G_{min}(t_1)$.

\end{definition}

\par

Put another way, if $u$ properly considers $v$ with respect to $[t_1, t_2]$,
then $u$ attempts to connect with $v$ in this interval with at least the same probability
as it would in a round of the synchronous algorithm corresponding to minimum productive subgraph $G_{min}(t_1)$.

It would simplify our analysis if for any time $t_1$ we could identify an interval $[t_1, t_2]$ such that $u$ properly considers $v$ for every edge $\{u,v\}\in E_{min}(t_1)$,
as we could then directly apply our analysis from the synchronous case.
We cannot, however, guarantee that such intervals always exist in our asynchronous setting.
Consider an edge $\{u,v\}\in E_{min}(t)$, for some $t$.
It might be the case that before $u$ can receive an advertisement from $v$,
that some other node connects to $v$ and transmits a token that removes $v$ from
the minimum productive subgraph. By the time $u$ subsequently hears from $v$,
it might no longer include it in its set of productive neighbors.

In some sense, however, this is a good case as it only {\em increases} the probability
that $v$ receives a connection attempt.
The following lemma (for which the proof can be found in Appendix \ref{proof:asyncrandom}) formalizes this intuition by proving that for any endpoint $v$ in a snapshot of the minimum
productive subgraph, $v$ will be selected with {\em at least} the probability that it would
if we had run a round of the synchronous algorithm on that snapshot.
This will allow us to subsequently apply Lemma~\ref{lem:prodconns}, 
which we carefully reworked from its original version in~\cite{ghaffari:2016}
so that it now only requires that this lower bound on selection probabilities holds.
(The original version made use of the exact selection probabilities from the synchronous algorithm.)

\begin{lemma}
\label{lem:asyncrandom}

Fix any time $t_1$ and node $v\in R_{min}(t_1)$, and let $N_{min}$ and $deg_{min}$ be the neighbor set and degree functions defined for $G_{min}(t_1)$. There exists a time $t_2$, where $t_2=t_1+\bigO(\delta_{max})$, such that $v$  connects productively in the interval $[t_1,t_2]$ with probability at least $1-\prod_{u\in N_{min}(v)}(1-1/deg_{min}(u))$.

\end{lemma}

With the above lemma, for any given time $t_1$, we have shown there is a time interval $[t_1,t_2]$ that is not too long such that each node in $L_{min}(t_1)$ behaves \emph{similarly} to the nodes in the synchronous setting with respect to $G_{min}(t_1)$. We now conclude by showing that this similarity is sufficient to apply the same analysis we used to prove Theorem \ref{thm:syncalg}.

\begin{proof}[Proof (of Theorem \ref{thm:asyncalg})]

The proof of our main theorem follows the same style of argument made by Lemmas \ref{lem:minmatching} through \ref{lem:finishsizeincrease} in our synchronous analysis.
Instead of assuming synchronized rounds, however, we now characterize our algorithm's behavior over
contiguous intervals of length $\ell = O(\delta_{max})$,
where $\ell$ is selected to be long enough to allow
Lemma~\ref{lem:asyncrandom} to apply to the intervals.

Let $t_i$ be the time at which interval $i$ begins.
We treat each interval $i$ like a round defined with respect to the minimum productive subgraph $G_{min}(t_i)$.
 The main difference in this new setting versus the synchronous is that Lemma~\ref{lem:asyncrandom}
 provides only a lower bound on a node in $R_{min}(t_1)$ being selected in interval $i$,
 whereas in the synchronous setting we know the exact probability of this event.
Fortunately, our reworked version of Lemma \ref{lem:prodconns} requires only this lower bound.
Indeed, much of the technical difficulty in reworking this lemma from its original form
was to allow it to require only this lower bound instead of precise probabilities.

\par

In more detail, the only property assumed of the algorithm by Lemma \ref{lem:prodconns} is that a node in $R_{min}(r)$ be selected with probability at least $1-\prod_{u\in N_{min}(v)}(1-1/deg_{min}(u))$. Since this is exactly what we showed in Lemma \ref{lem:asyncrandom} for our asynchronous algorithm, Lemma \ref{lem:prodconns} applies to the graphs corresponding to our intervals.
The remainder of the relevant lemmas in our synchronous analysis require only
that Lemma \ref{lem:prodconns} holds.
We can therefore apply these lemmas to our intervals to obtain a similar complexity
bound for gossiping $k$ tokens, with the only difference being that instead of bounding
the number of rounds, we bound the number of intervals of length $O(\delta_{max})$
that are required.
\end{proof}

%% file: appendix.tex
\subsection{Omitted Proofs from Section \ref{sec:prelim}}

\subsubsection{Proof of Lemma \ref{lem:stochasticdominance}}
\label{proof:stochasticdominance}

For each $i\in [T]$, let $\hat{X}_i$ be an indicator variable that is $1$ with probability $p$ and $0$ otherwise. We now define a process for generating a \textit{coupled} distribution where the values sampled are pairs of bits. Namely, we sample $i$ pairs $(Y_i, \hat{Y}_i)$ where $\hat{Y}_i$ is $1$ with probability $p$ and $0$ otherwise. If $\hat{Y}_i$ is $1$, we set $Y_i$ to 1 as well. Otherwise, we set $Y_i$ to $1$ with probability $(q_i-p)/(1-p)$. In this way, the marginal probability that $Y_i=1$ is exactly $q_i$, the same success probability as our original indicator variable $X_i$. Clearly $\hat{X}_i$ and $\hat{Y}_i$ are also $1$ with the same probability $p$. Therefore, since for any $T$-sequence execution where  $\hat{Y}=\sum_{i=1}^T\hat{Y}_i$  and $Y=\sum_{i=1}^TY_i$ it's true that $Y\geq \hat{Y}$, it follows that as long as $\hat{Y}$ is at least some value, then so is $Y$. In other words, it's sufficient to lower bound the value of $\hat{Y}$ to derive a lower bound for the number of successes in the the series $X_1,\ldots,X_T$.

\par

For the expectation $\E[\hat{Y}]=pT$, we can apply the Chernoff bound from Theorem $\ref{thm:chernoff}$ with $\varepsilon=1/2$ to upper bound the probability $\hat{Y}$ is less than $T/2$,

\begin{align*}
    &\Pr[\hat{Y}\leq(1/2)\cdot pT]=\Pr[\hat{Y}\leq\Theta(pT)]\\
    &\leq\exp\bigg(-\frac{\Theta(pT)}{8}\bigg)=\bigO(\exp(-pT))
\end{align*}
Therefore, with very high probability in $pT$, $Y\geq \hat{Y}=\Omega(pT)$.

\subsection{Omitted Proofs from Section \ref{sec:synch}}
\subsubsection{Proof of Lemma \ref{lem:samesize}}
\label{proof:samesize}

\begin{proof}

When $C(r) = 1$, all nodes $u\in V$ have the same number of tokens, therefore $s_u(r)=|T_v(r)|$ for all edges $\{u,v\}\in E$. Clearly, if not all nodes have $k$ tokens, then there is some node $u$ with some neighbor $v$ such that $H(T_u)\neq H(T_v)$. Furthermore, since $s_u(r)=|T_v(r)|$ for all edges $\{u,v\}\in E$, $v\in N_u(r)$. Therefore, since $|N_u(r)|>0$, $u$ will send a connection proposal to some neighbor $w\in N_u(r)$.

\par

If $w$ receives a connection proposal from $u$, $w$ is guaranteed to accept at least one connection proposal this round and participate in at least one connection this round (initiated by a neighbor's proposal). Therefore, after this round $w$ will possess a new token such that $|T_w(r+1)|>i^*(r)$. There are now two possibilities: $\forall u\in V, |T_u(r+1)|\geq |T_w(r+1)|$ or $\exists u\in V, |T_u(r+1)|< |T_w(r+1)|$. In the first case, $i_{min}(r+1)>i_{min}(r)$ since all nodes possess more than $i_{min}(r)$ tokens. In the second case $C(r+1)>1$ since the chosen node $u$ has fewer tokens than $w$.
\end{proof}

\subsubsection{Proof of Lemma \ref{lem:minmatching}}
\label{proof:minmatching}

Fix the cut $(S_{min}(r), V\setminus S_{min}(r))$ and recall that $n^*_{min}(r)$ is the size of the smaller of the two bipartitions. From Lemma \ref{lem:matching} we know that there is a matching $M$ of size at least $(\alpha/4)\cdot n^*_{min}(r)$ across this cut. Consider an arbitrary edge $\{u,v\}\in M$ and without loss of generality assume $u\in V\setminus S_{min}(r)$ and $v\in S_{min}(r)$.

\par

By the definition of $S_{min}(r)$, no node has fewer tokens than $v$ and therefore $v$ must have the smallest token set size out of all $u$'s neighbors. Since $v\in N(u)$, this means that $s_u(r)=|T_v(r)|$ which implies that $H(T_u(r))\neq H(T_v(r))$. This is sufficient to show that $\{u,v\}$ is in the productive subgraph.

\par

Furthermore, clearly $N(u)\cap S_{min}(r)\neq\emptyset$ and $N(v)\cap (V\setminus S_{min}(r))\neq\emptyset$ and therefore $u\in L_{min}(r)$ and $v\in R_{min}(r)$. Therefore, it must be the case that $\{u,v\}$ is in the minimum productive subgraph as well. Since we can show this for any arbitrary edge $\{u,v\}\in M$, it is true for every such edge in the matching. Therefore $M$ is also a matching of size at least $(\alpha/4)\cdot n^*_{min}(r)$ over $G_{min}(r)$.

\subsubsection{Proof of Lemma \ref{lem:fullprogress}}
\label{proof:fullprogress}

We can now apply the same reasoning as the proof of Theorem 7.2 in \cite{ghaffari:2016} to show that applying Lemma \ref{lem:prodconns} inductively over $\bigO(\log{\Delta})$ rounds achieves our desired number of connections. We summarize this argument here.

\par

Fix a round $r$ and apply the first iteration of Lemma \ref{lem:prodconns} at the beginning of this round. For the $i$th application of Lemma \ref{lem:prodconns}, let $m_i$ be the size of the maximum matching over $G(L,R)$ in this iteration. By the lemma statement, with at least constant probability, either $\Omega(m_i/\log{\Delta})$ nodes of $R_{min}(r)$ participate in a productive connection or we can identify some subgraph $G(L'',R'')$ defined according to the second lemma objective. If the latter holds, this graph $G(L'',R'')$ becomes the new $G(L,R)$ for the $(i+1)$th iteration.

\par

As in \cite{ghaffari:2016}, when applying our lemma inductively we must ensure that this new value for $G(L,R)$ in the $(i+1)$th iteration satisfies the criteria required by the lemma setup. Clearly criteria $1$, $2$, and $4$ are satisfied by properties $a$, $b$, and $d$ of objective 2 of Lemma \ref{lem:prodconns}, respectively. The only property that is not trivially satisfied is therefore $3$. To ensure that $m_i$ is not too small compared to the size $m$ of our original matching over $G(L,R)$, we note that for all $i$, $m_i\geq (1-1/\log{\Delta})^{2i}\cdot m\geq c\cdot m$. Therefore, since $i\leq32\cdot \log{\Delta}$, this expression is made valid by setting $c = \exp(-64)$ (please note that while we make no effort to do so here in lieu of a clearer probabilistic analysis, this constant can certainly be optimized). Therefore, $m_i=\Omega(m)$ for any inductive step $i$ and so if the first objective of the lemma is satisfied on any iteration, at least $\Omega(m/\log{\Delta})$ nodes in $R_{min}(r)$ participate in a productive connection. 

\par

Furthermore notice that after $32\cdot\log{\Delta}$ steps where the second objective of the lemma is satisfied, the degree sum of the final graph is $m_i$. Therefore each node in $L$ only has one neighbor to choose from such that the number of productive connections with nodes in $R_{min}(r)$ is trivially $m_i=\Omega(m_i/\log{\Delta})$. Therefore, it holds that after $32\cdot\log{\Delta}$ steps in which at least one of the lemma objectives is satisfied, at least $\Omega(m/\log{\Delta})$ nodes in $R_{min}(r)$ participate in a productive connection (again where $m$ is the maximum size of the matching over $G_{min}(r)$).

\par

Call any round where at least one of the objectives of Lemma \ref{lem:prodconns} is satisfied a success. Since we know each round is successful with at least constant probability, we can apply the stochastic dominance argument from Lemma \ref{lem:stochasticdominance} to demonstrate (with high probability in $\Delta$) that at most $\bigO(\log{\Delta})$ steps are required to achieve $32\cdot\log{\Delta}$ successes. Therefore at most $\bigO(\log{\Delta})$ rounds are required before $\Omega(m/\log{\Delta})$ nodes in $R_{min}(r)$ participate in a productive connection.

\par

Finally, from Lemma \ref{lem:minmatching} we know that for any fixed round $r$ there is a matching over the minimum productive subgraph of size $(\alpha/4)\cdot n^*_{min}(r)$. Therefore, $m\geq (\alpha/4)\cdot n^*_{min}(r)$ and so at least $\Omega((\alpha/\log{\Delta})\cdot n^*_{min}(r))$ nodes of $R_{min}(r)\subseteq S_{min}(r)$ participate in a productive connection after at most $\bigO(\log{\Delta})$ rounds.
\subsubsection{Helper Lemma to Support Lemma \ref{lem:finishsizeincrease}}

\begin{lemma}
\label{lem:sizeincrease}

Fix a round $r>0$ such that $n-n_{min}(r)\leq n/2$. With high probability in $n$, after at most $\bigO((1/\alpha)\log{n}\log^2{\Delta})$ rounds, more than half of the nodes posses more than $i_{min}(r)$ tokens.
\end{lemma}

Call a phase $p_i$ of $\bigO(\log{\Delta})$ rounds \emph{successful} if at least $\bigO((\alpha/\log{\Delta})\cdot n_{min}^*(r_i))$ nodes of $S_{min}(r_i)$ participate in a productive connection where $r_i$ is the first round of $p_i$. We observe $t$ successful phases  $p_1,\ldots,p_t$ while at most half the nodes have more than the minimum token set size. Our goal will be to show that there can only be so many of these phases before at least half the nodes in the network possess more than $i_{min}(r_1)$ tokens.

\par

Notice that for any integers $i$ and $j$ such that $i<j$, $n-n_{min}(r_i)\leq n-n_{min}(r_j)$, since we can only \emph{grow} the number of nodes with more than the minimum number of tokens. Furthermore, recall that for any $i$ such that $n_{min}(r_i)\leq n/2$, by definition $n_{min}(r_i)=n^*_{min}(r_i)$. Therefore, if $m_i$ is the size of the maximum matching over $G_{min}(r_i)$, for any $i$ and $j$ where $i<j$ and $n^*_{min}(r_i)\leq n^*_{min}(r_j)\leq n/2$ we know that $m_i\leq m_j$. Therefore (by Lemma \ref{lem:minmatching}) we have that the $i$th successful phase achieves $\Omega((\alpha/\log{\Delta})\cdot n^*_{min}(r_i))\geq\Omega((\alpha/\log{\Delta})\cdot n^*_{min}(r_1))$ productive connections. 

\par

Therefore at most $\bigO(\log{\Delta}/\alpha)$ successful phases are required until the number nodes with more than the minimum number of tokens grows by a constant fraction. We can now group together $T$ sequences of $\bigO(\log{\Delta}/\alpha)$ phases $s_1,\ldots,s_T$ and solve for $T$ such that
\begin{align*}
    n^*_{min}(r)\cdot (1+\Omega(1))^T&\geq n/2\\
\end{align*}

Which yields $T\leq \bigO(\log{n})$ sequences for a total of $\bigO((1/\alpha)\log{n}\log{\Delta})$ total phases. Finally, we bound how many of these phases must pass until until we achieve a sufficient number of successful phases. To this end, we apply the stochastic dominance argument from Lemma \ref{lem:stochasticdominance} using the constant probability lower bound introduced by Lemma \ref{lem:prodconns}. This gives us with high probability in $n$ that our final phase complexity is $\bigO((1/\alpha)\log{n}\log{\Delta})$ phases and our final round complexity is $\bigO((1/\alpha)\log{n}\log^2{\Delta})$ total rounds.

\subsubsection{Proof of Lemma \ref{lem:finishsizeincrease}}
\label{proof:finishsizeincrease}

In Lemma \ref{lem:sizeincrease} we showed that when $n-n_{min}(r)\leq n/2$ for some round $r>0$, by periodically growing the set of nodes with more than the fewest number of tokens by a constant fraction, at most $\bigO((1/\alpha)\log{n}\log^2{\Delta}$ rounds are required until $n_{min}(r)\leq n/2$. A symmetric argument can be made in the case that $n_{min}(r)\leq n/2$ in which we \textit{shrink} the nodes with the fewest number of tokens by a constant fraction until at most a constant number remain.

\par

Let $s_1,\ldots s_T$ be several sequences where $s_i=p_1,\ldots,p_t$ is made up of $\bigO(\alpha/\log{\Delta})$ phases where each phase $p_i$ is made up of $\bigO(\log{\Delta})$ rounds. While in Lemma \ref{lem:sizeincrease} we lower bounded the size of the matching in each phase in each sequence by the first round of $p_1$ of $s_1$, we now lower bound the size of these matchings by the last round of $p_t$ of $s_T$. The maximum matching over the minimum productive subgraph with respect to this round gives us a lower bound on the number of connections achieved in each successful phase up to this round.

\par 

Therefore, in order to solve for $T$ in this case we solve the expression $(n-n_{min}(r))\cdot (1-\Omega(1))^T\leq 1$ for our initial round $r$. Again, this expression indicates for the $i$th sequence we reduce the number of nodes in $S_{min}(r)$ by a constant fraction of $n^*_{min}(r')$ where $r'$ is the round at the beginning of the $(i+1)$th sequence. This again yields $T\leq \bigO(\log{n})$. The rest of the proof is the same as that of Lemma \ref{lem:sizeincrease}, yielding a final round complexity of $\bigO((1/\alpha)\log{n}\log^2{\Delta})$ total rounds to inform all but one node which is then trivially connected to in at most one more round.

\subsubsection{Proof of Lemma \ref{thm:gossiplowerbound}}
\label{proof:gossiplowerbound}
Let $q=n\alpha$. Construct our graph $G$ with vertex expansion at least $\alpha$ by creating a $q$-clique of nodes and connecting the remaining $n-q$ nodes to every node in the $q$-clique. This graph is equivalent to a star graph when $q=1$ ($\alpha=1/n$) and to a clique when $q=n$ ($\alpha=1$).

\par

First, we'll show that this graph has vertex expansion at least $\alpha$. Recall that to find the vertex expansion of a graph, the goal is to minimize the quantity $|\partial(S)|/|S|$ over all cuts $S$ of size at most $n/2$. When considering all possible cuts of our graph, we can choose to either include nodes from the $q$-clique or nodes not in the clique (or both). If we select any node from the $q$-clique, by our construction, every remaining node is now in $\partial(S)$. Therefore, the only freedom we have to minimize $|\partial(S)|/|S|$ is to increase the size of $S$ as to maximize the denominator. However, the minimum value we can derive is still only $1\geq 2\alpha$ since we can include at most $n/2$ nodes in $S$.

\par

Our only remaining option is then to try to minimize $|\partial(S)|/|S|$ by not including any nodes from the $q$-clique in our set $S$. As soon as we include a single node outside the $q$-clique, $|\partial(S)|=q$. Furthermore, we are compelled to include up to $n/2$ nodes this way (since $q\leq 1/2$ so there are at least $n/2$ nodes not in the $q$-clique). This minimizes the target quantity since it has no effect on the numerator and it maximizes the denominator. However, even when the quantity is minimized in this way, it's always the case that $|\partial(S)|/|S|\geq q/(n/2)$. Lastly, since $q=n\alpha$, this quantity is at least $2\alpha$, satisfying the condition of the Lemma statement.

\par

We will now show that it takes $\Omega(k/\alpha)$ rounds to spread $k$ tokens to all $n$ nodes in $G$. We begin by providing all $k$ tokens to every node in the $q$-clique. To solve the gossip problem, all $k$ tokens must be delivered to the $n-q\geq n/2$ nodes not in the clique. This requires at least $kn/2$ total connections to be made. However, since at most $q$ connections can occur per round (since nodes outside the clique aren't connected and the nodes in the clique are limited to at most one connection per round) a total of at least $kn/(2q)$ rounds are required. Substituting $q=n\alpha$ then gives the needed lower bound on the total number of rounds: $\Omega(kn/(2q))=\Omega(kn/(2n\alpha))=\Omega(k/(2\alpha))=\Omega(k/\alpha).$

\subsection{Omitted Proofs from Section \ref{sec:asynch}}
\subsubsection{Proof of Lemma \ref{lem:asyncrandom}}
\label{proof:asyncrandom}

Fix any time $t_1$ and node $v$ as specified in the lemma statement.
Fix $t_2$ to be the minimum time after $t_1$ that is sufficiently large
to guarantee that for every pair of neighbors $\{x,y\}$ in the underlying network topology,
$y$ receives an advertisement from $x$ that was passed to $x$'s {\tt update} method
at some time $t \geq t_1$, and $y$ has time to call \textsc{Select}  after receiving 
at least one such an advertisement.
Clearly, $t_2 = t_1 + O(\delta_{max})$.

We will consider all possible executions of our algorithm over the interval $[t_1, t_2]$. 
We partition these executions into two disjoint event spaces with respect to $v$. The first space, which we will denote $\mathcal{A}$, will contain all executions in which every node $u\in N_{min}(v)$ properly considers $v$ during the interval $[t_1, t_2]$.
(Recall that in the lemma statement we define $N_{min}(v)$ to be the neighbor set of $v$ in $G_{min}(t_1)$, and $deg_{min}(v)=|N_{min}(v)|$.) 
The second space  $\overline{\mathcal{A}}$ will then simply be the complement of $\mathcal{A}$, containing all other executions.

\par

Begin with some execution $a\in \mathcal{A}$. Recall by the definition of $\mathcal{A}$, every node in $N_{min}(v)$ properly considers $v$ during the interval $[t_1, t_2]$ in execution $a$. 
Fix one such neighbor $u\in N_{min}(v)$.
There is some time $t_3$ in our interval
such that at this time, $u$ makes a call to $\textsc{Select}$,
during which $v\in \hat{N}_u(t_3)$ and $|\hat{N}_u(t_3)|\leq deg_{min}(u)$. 
During this call, $u$ will select $v$ for a connection attempt
 with probability $1/|\hat{N}_u(t_3)|\geq 1/deg_{min}(u)$.
 It follows that the probability that $u$ {\em does not} send $v$ a proposal
 is at most $1-1/deg_{min}(u)$.

\par

Since execution $a$ is in the event space $\mathcal{A}$, we know that every node $u\in N_{min}(v)$ properly considers $v$ during this interval. Moreover, given that these nodes properly consider $v$, the probability that two nodes return $v$ from $\textsc{Select}$ is independent (since the selected neighbor is chosen uniformly at random from $\hat{N}$). Therefore, we can bound the probability of event $X$, where $X$ denotes that $v$ receives at least one connection proposal, as follows:
\begin{align*}
\Pr[\neg X]&\leq \Pi_{u\in N_{min}(v)}(1-deg_{min}(u))\\
\Pr[X]&\geq 1-\Pi_{u\in N_{min}(v)}(1-deg_{min}(u))
\end{align*}

Furthermore, by the guarantees of the aMTM, having received at least one connection proposal, $v$ is guaranteed to accept at least one. Therefore, $v$ participates in a productive connection with at least the above probability. 

We now consider some execution $a\in \overline{\mathcal{A}}$. 
By the definition of $\overline{\mathcal{A}}$, there must be in $a$ some node $u\in N_{min}(v)$, such that $u$ does not properly consider $v$ in the interval $[t_1, t_2]$.
Fix $t_{select}$ to be the first time that $u$ calls \textsc{Select}
after receiving an advertisement from $v$ that was passed to {\tt update} at a time greater
than or equal to $t_1$.
By our definition of $t_2$, we can always identify a time $t_{select}$ that satisfies 
these properties in $[t_1, t_2]$.

By assumption,
we know that $u$ does not properly consider $v$ during the call to \textsc{Select} at $t_{select}$.
We consider the two possible reasons for this behavior,
and show in both cases $v$ must have already participated in a productive connection between $t_1$ and $t_{select}$.

The first possible reason is that $v\notin \hat{N}_u(t_{select})$.
By the definition of our algorithm, the minimum token set size in the network can never {\em decrease}.
It follows that if $v\notin \hat{N}_u(t_{select})$, then $v$ must have learned at least one
token since $t_1$. It follows that $v$ participated in a productive connection since $t_1$.

The second reason that $u$ might not properly consider $v$ would be if $v\in \hat{N}_u(t_{select})$,
but $\hat{N}_u(t_{select})$ is too large such that $|\hat{N}_u(t_{select})| > deg_{min}(u)$.
However, at time $t_1$, exactly $deg_{min}(u)$ neighbors of $u$ had a token set size of $i_{min}(t_1)$ (by the definition of $G_{min}(t_1)$), with all other neighbors of $u$ in $G$ having strictly more tokens. Since 
nodes cannot lose tokens, the number of $u$'s neighbors with at most $i_{min}(t_1)$ tokens can never increase.
If $|\hat{N}_u(t_{select})| > deg_{min}(u)$,
then $\hat{s}_u(t_{select})>i_{min}(t_1)$,
from which it follows that $v$, along with all of $u$'s neighbors with token set size
$i_{min}(t_1)$ at $t_1$, must have received at least one token since $t_1$,
meaning it participated in a productive connection.

\par

We have shown, therefore,
that for any $a\in \overline{\mathcal{A}}$, 
$v$ participates in a productive connection in $[t_1, t_2]$ in $a$ with probability $1$.
Pulling together these pieces,
we have partitioned the possible executions in the interval $[t_1, t_2]$ into two sets.
In both sets, the probability of $v$ participating in a productive connection
is at least $1-\prod_{u\in N_{min}(v)}(1-deg_{min}(u))$, as required by the lemma statement.

\subsection{Swift Implementation of the Asynchronous Random Diffusion Gossip Algorithm}
\label{apx:implementation}

\begin{Figure}
\label{fig:swiftimplementation}

\begin{lstlisting}[language=swift]
init() {
  self.tokens = Set<String>()
  // standard implementation
  self.hash = Utils.SHA256
}

func getTag() -> Advertisement {
  let tokenSetHash = self.hash(tokens)
  let tokenSetSize = self.tokens.count
  // iOS device identifier
  let uid = UIDevice.current.name
  return Advertisement(tokenSetHash
    , tokenSetSize, uid)
}

func select(advertisements
  : [Advertisement])
    -> String {
  let tokenSetHash = self.getTag()
    .tokenSetHash
  
  var minSize = Int.max
  for adv in advertisements {
    let size = adv.tokenSetSize!
    if  (size < minSize)
      && (adv.tokenSetHash
        != tokenSetHash) {
      minSize = size
    }
  }
  
  var neighbors = Set<String>()
  for adv in advertisements {
    if ((adv.tokenSetSize! == minSize) 
      && (adv.tokenSetHash
        != tokenSetHash)) {
      neighbors.insert(adv.uid)
    }
  }
  // returns device uid
  return neighbors.randomElement()!
}

func communicate(neighborTokens:
  Set<String>) {
  let difference = neighborTokens
    .subtracting(self.tokens)
  self.tokens
    .insert(difference.randomElement()!)
}
\end{lstlisting}
\end{Figure}

%% file: main.bbl
\begin{thebibliography}{10}

\bibitem{aloi2014spontaneous}
Gianluca Aloi, Marco Di~Felice, Valeria Loscr{\`\i}, Pasquale Pace, and
  Giuseppe Ruggeri.
\newblock Spontaneous smartphone networks as a user-centric solution for the
  future internet.
\newblock {\em IEEE Communications Magazine}, 52:26--33, 2014.

\bibitem{alon:1986}
Noga Alon, L\'{a}zsl\'{o} Babai, and Alon Itai.
\newblock A fast and simple randomized parallel algorithm for the maximal
  independent set problem.
\newblock {\em Journal of Algorithms}, 7:567--583, 1986.

\bibitem{barry:2012}
Peter Barry and Patrick Crowley.
\newblock {\em Modern Embedded Computing: Designing Connected, Pervasive,
  Media-Rich Systems}.
\newblock Morgan Kaufmann Publishers Inc., San Francisco, CA, USA, 1st edition,
  2012.

\bibitem{censorhillel2017}
Keren Censor-Hillel, Bernhard Haeupler, Jonathan~A. Kelner, and Petar
  Maymounkov.
\newblock Rumor spreading with no dependence on conductance.
\newblock {\em SIAM J. Comput.}, 46:58--79, 2017.

\bibitem{chierichetti2010rumour}
Flavio Chierichetti, Silvio Lattanzi, and Alessandro Panconesi.
\newblock Rumour spreading and graph conductance.
\newblock In {\em Proceedings of the Symposium on Discrete Algorithms (SODA)},
  pages 1657--1663. SIAM, 2010.

\bibitem{kuhn:bounded}
Sebastian Daum, Fabian Kuhn, and Yannic Maus.
\newblock Rumor spreading with bounded in-degree.
\newblock In {\em International Colloquium on Structural Information and
  Communication Complexity (SIRROCO)}, pages 323--339. Springer International
  Publishing, 2016.

\bibitem{dinitz:2017}
Michael Dinitz, Jeremy Fineman, Seth Gilbert, and Calvin Newport.
\newblock Load balancing with bounded convergence in dynamic networks.
\newblock In {\em Proceedings of the of the International Conference on
  Computer Communications (INFOCOM)}, pages 1--9. IEEE, 2017.

\bibitem{dinitz2019capacity}
Michael Dinitz, Magn{\'u}s~M Halld{\'o}rsson, Calvin Newport, and Alex Weaver.
\newblock The capacity of smartphone peer-to-peer networks.
\newblock In {\em Proceedings of the International Symposium on Distributed
  Computing (DISC)}, pages 14:1--14:17. Schloss Dagstuhl, 2019.

\bibitem{fountoulakis2010rumor}
Nikolaos Fountoulakis and Konstantinos Panagiotou.
\newblock Rumor spreading on random regular graphs and expanders.
\newblock In {\em Proceedings of the International Conference on Approximation,
  and the International Conference on Randomization, and Combinatorial
  Optimization: Algorithms and Techniques (APPROX/RANDOM)}, pages 560--573.
  Springer-Verlag, 2010.

\bibitem{telephone1}
Alan~M Frieze and Geoffrey~R Grimmett.
\newblock The shortest-path problem for graphs with random arc-lengths.
\newblock {\em Discrete Applied Mathematics}, 10:57--77, 1985.

\bibitem{firechat}
Open Garden.
\newblock Firechat, 2018.
\newblock URL: \url{https://www.opengarden.com}.

\bibitem{oghostpot}
Open Garden.
\newblock The open garden hotspot, 2018.
\newblock URL: \url{https://www.opengarden.com}.

\bibitem{ghaffari:2016}
Mohsen Ghaffari and Calvin Newport.
\newblock How to discreetly spread a rumor in a crowd.
\newblock In {\em Proceedings of the International Symposium on Distributed
  Computing (DISC)}, pages 357--370, 2016.

\bibitem{telephone2}
George Giakkoupis.
\newblock Tight bounds for rumor spreading in graphs of a given conductance.
\newblock In {\em Proceedings of the Symposium on Theoretical Aspects of
  Computer Science (STACS)}, pages 57--68, 2011.

\bibitem{giakkoupis2014tight}
George Giakkoupis.
\newblock Tight bounds for rumor spreading with vertex expansion.
\newblock In {\em Proceedings of the Symposium on Discrete Algorithms (SODA)},
  pages 801--815. SIAM, 2014.

\bibitem{telephone3}
George Giakkoupis and Thomas Sauerwald.
\newblock Rumor spreading and vertex expansion.
\newblock In {\em Proceedings of the Symposium on Discrete Algorithms (SODA)},
  pages 1623--1641. SIAM, 2012.

\bibitem{holzer2016padoc}
Adrian Holzer, Sven Reber, Jonny Quarta, Jorge Mazuze, and Denis Gillet.
\newblock Padoc: Enabling social networking in proximity.
\newblock {\em Computer Networks}, 111:82--92, 2016.

\bibitem{lu2016networking}
Zongqing Lu, Guohong Cao, and Thomas La~Porta.
\newblock Networking smartphones for disaster recovery.
\newblock In {\em Proceedings of the International Conference on Pervasive
  Computing and Communications (PerCom)}, pages 1--9. IEEE, 2016.

\bibitem{newport:2017b}
Calvin Newport.
\newblock Gossip in a smartphone peer-to-peer network.
\newblock In {\em Proceedings of the Symposium on Principles of Distributed
  Computing (PODC)}, pages 43--52. ACM, 2017.

\bibitem{newport:2017}
Calvin Newport.
\newblock Leader election in a smartphone peer-to-peer network.
\newblock In {\em Proceedings of the International Parallel and Distributed
  Processing Symposium (IPDPS)}, pages 172--181. IEEE, 2017.

\bibitem{newport2019random}
Calvin Newport and Alex Weaver.
\newblock Random gossip processes in smartphone peer-to-peer networks.
\newblock In {\em Proceedings of the International Conference on Distributed
  Computing in Sensor Systems (DCOSS)}, pages 139--146. IEEE, 2019.

\bibitem{reina2015survey}
DG~Reina, Mohamed Askalani, SL~Toral, Federico Barrero, Eleana Asimakopoulou,
  and Nik Bessis.
\newblock A survey on multihop ad hoc networks for disaster response scenarios.
\newblock {\em International Journal of Distributed Sensor Networks},
  11:647037, 2015.

\bibitem{suzuki2012soscast}
Noriyuki Suzuki, Jane Louie~Fresco Zamora, Shigeru Kashihara, and Suguru
  Yamaguchi.
\newblock Soscast: Location estimation of immobilized persons through sos
  message propagation.
\newblock In {\em Proceedings of the International Conference on Intelligent
  Networking and Collaborative Systems (INCoS)}, pages 428--435. IEEE, 2012.

\end{thebibliography}
